%% file: document.tex
\documentclass[a4paper,10pt]{article}
\usepackage{fullpage}
\usepackage{graphicx}
\usepackage[caption=false]{subfig}
\usepackage{booktabs}
\usepackage{url}
\usepackage{microtype}
\usepackage{algorithm}
\usepackage{algpseudocode}
\usepackage[shortcuts]{extdash}
\usepackage[bm]{pauli_all}
\usepackage{xspace}
\usepackage{color}

\newtheorem{proposition}{Proposition}[section]

\newcommand{\outprod}{\ensuremath{\boxtimes}}
\newcommand{\kprod}{\ensuremath{\otimes}}
\newcommand{\krprod}{\ensuremath{\odot}}

\newcommand{\E}{\mathbf{E}}

\newcommand{\RDF}{\textsc{rdf}\xspace}
\newcommand{\SPARQL}{\textsc{sparql}\xspace}

\newlength{\figwidth}
\newlength{\figsep}
\newcommand{\todo}[1]{\par{\raggedright\color{red} #1}\par}
\renewcommand\todo[1]{}

\begin{document}


\title{On Defining SPARQL with Boolean Tensor Algebra}
\author{Saskia Metzler\\
  Max-Planck-Institut f\"ur Informatik\\
  Saarbr\"ucken, Germany\\
  \texttt{saskia.metzler@mpi-inf.mpg.de}
\and
Pauli Miettinen\\
  Max-Planck-Institut f\"ur Informatik\\
  Saarbr\"ucken, Germany\\
  \texttt{pauli.miettinen@mpi-inf.mpg.de}
}

\maketitle
\begin{abstract}
The Resource Description Framework (\RDF) represents information as subject--predicate--object triples. These triples are commonly interpreted as a directed labelled graph. We propose an alternative approach, interpreting the data as a 3-way Boolean tensor. We show how \SPARQL queries -- the standard queries for \RDF\ -- can be expressed as elementary operations in Boolean algebra, giving us a complete re-interpretation of \RDF and \SPARQL. We show how the Boolean tensor interpretation allows for new optimizations and analyses of the complexity of \SPARQL queries. For example, estimating the size of the results for different join queries becomes much simpler.
\end{abstract}

\input{introduction}
\input{notation}
\input{data}
\input{operations}
\input{cardinality}

\input{decompositions}
\input{related}
\input{discussion}

\input{conclusions}

\bibliographystyle{abbrv}
\bibliography{library}  
\end{document}

%% file: introduction.tex
\section{Introduction}
\label{sec:introduction}
The Resource Description Framework (\RDF)\footnote{\url{http://www.w3.org/TR/rdf-syntax-grammar}} is a W3C standard for representing information in the web. \RDF data consist of subject-predicate-object $(s, p, o)$ triples that are commonly treated as a directed labelled graph. Edges in the \RDF graph go from subjects to objects and predicates form the edge labels.  The \RDF data can be queried with the \SPARQL query language\footnote{\url{http://www.w3.org/TR/rdf-sparql-query}}. The labelled directed graph interpretation of \RDF data allows for defining some \SPARQL operations in an intuitive way, but it is not always the most convenient theoretical framework to work with: for example, the $(s, p, o)$ triples treat the predicate in no different way to the subject or object, yet in the graph interpretation, the predicate acts as the edge, while the subject and object are nodes. 

In this paper we propose to approach the \RDF data as a 3-way Boolean tensor, defining the \SPARQL operations using elementary tensor and matrix operations. While seeing \RDF data as tensor is nothing new, to the best of our knowledge we present the first comprehensive analysis of \SPARQL in terms of Boolean tensor operations (although there are prior work on efficiently implementing specific \SPARQL operations using binary tensors, see Section~\ref{sec:relatedWork}). 

We want to emphasize that we do not propose that efficient \RDF databases should be build on top of the tensor interpretation. We do think, however, that the tensor interpretation makes certain optimization techniques more intuitive than the graph interpretation. This will hopefully yield concrete benefits for the query processing, for example, by providing more efficient cardinality estimators for joins.

After a brief introduction to tensor notation and terminology (Section~\ref{sec:notation}), we will explain the data model and how the \SPARQL queries can be evaluated directly using tensor algebra (Sections~\ref{sec:dataModel} and~\ref{sec:operations}). For the sake of clarity, we present the correspondence between \SPARQL and tensor algebra using examples. We will then start studying the results we can get using our tensor formulation. First we will study the computation of joins and the estimation of their cardinalities (Section~\ref{sec:cardinality}), before moving to tensor decompositions (Section~\ref{sec:decompositions}) and their properties. Throughout these two sections, we list a number of propositions. Many of them are either straight forward, or have been proved earlier. Our goal in presenting these results is to show what kind of results our framework facilitates. We will cover some related work, interesting future directions, and conclusions in the last three sections.


%% file: notation.tex
\section{Notation}
\label{sec:notation}
Throughout this paper, we indicate vectors as bold type lower\-/case letters ($\vec{v}$), matrices as bold upper\-/case letters ($\matr{M}$), and tensors as bold\ upper\-/case calligraphic letters ($\tens{T}$). 

Element $(i,j,k)$ of a 3-way tensor $\tens{X}$ is denoted as $x_{ijk}$. A colon in a subscript denotes taking that mode entirely; for example, $\matr{X}_{::k}$ is the $k$th \emph{frontal slice} of $\tens{X}$ ($\matr{X}_k$ in short), and $\vec{x}_{:jk}$ is a \emph{tube} along the first way of $\tens{X}$. 

For a 3-way tensor \tens{X}, $x_{:jk}$ is the \emph{mode-1 (row) fibre}, $x_{i:k}$ is the \emph{mode-2 (column) fibre}, and $x_{ij:}$ is the \emph{mode-3 (tube) fibre}. Furthermore, $\matr{X}_{::k}$ is the $k$th frontal slice of \tens{X}.

A tensor can be \emph{unfolded} into a matrix by arranging its fibres as columns of a matrix. 
The mode-$i$ \emph{matricization}, of \byby{n}{m}{l} tensor $\tens{T}$, denoted $\matr{T}_{(n)}$, takes the mode-$i$ fibres of $\tens{T}$ and arranges them as the columns of matrix $\matr{T}_{(i)}$. For example, in mode-1 unfolding, the columns of $\tens{T}$ constitute the columns of $\matr{T}_{(1)}$ that has $n$ rows and $ml$ columns.

The outer product of vectors in $N$ modes is denoted by \outprod. That is, if \vec{a}, \vec{b}, and \vec{c} are vectors of length $m$, $n$, and $l$, respectively, $\tens X = \vec a \outprod \vec b \outprod \vec c$ is an \byby{m}{n}{l} tensor with $x_{ijk}=a_ib_jc_k$.

The \emph{Boolean tensor sum} of binary tensors $\tens{X}$ and $\tens{Y}$ is defined as $(\tens{X}\lor\tens{Y})_{ijk} = x_{ijk}\lor y_{ijk}$.
For binary matrices $\matr{X}$ and $\matr{Y}$ where $\matr{X}$ has $r$ columns and $\matr{Y}$ has $r$ rows their \emph{Boolean matrix product}, $\matr{X}\bprod\matr{Y}$, is defined as $(\matr{X}\bprod\matr{Y})_{ij} = \bigvee_{k=1}^r x_{ik}y_{kj}$. 

Let $\matr{X}$ be an \by{n_1}{m_1} matrix  and $\matr{Y}$ be an \by{n_2}{m_2} matrix. Their \emph{Kronecker (matrix) product} is the \by{n_1n_2}{m_1m_2} matrix $\matr{X}\kprod\matr{Y}$ defined by
\begin{equation}
\label{eq:kronecker}
\matr{X}\kprod\matr{Y}=
\left(\begin{matrix}
  x_{11}\matr{Y} & x_{12}\matr{Y} & \cdots & x_{1m_1}\matr{Y}\\
  x_{21}\matr{Y} & x_{22}\matr{Y} & \cdots & x_{2m_1}\matr{Y}\\
  \vdots            &    \vdots        & \ddots & \vdots              \\
  x_{n_11}\matr{Y} & x_{n_12}\matr{Y} & \cdots & x_{n_1m_1}\matr{Y} 
\end{matrix}\right).
\end{equation}

The \emph{Khatri--Rao (matrix) product} of $\matr{X}$ and $\matr{Y}$ is defined as ``column-wise Kronecker''. That is, $\matr{X}$ and $\matr{Y}$ must have the same number of columns ($m_1=m_2=m$), and their Khatri--Rao product $\matr{X}\krprod\matr{Y}$ is the \by{n_1n_2}{m} matrix defined as
\begin{align}\nonumber
\label{eq:khatri-rao}
\matr{X}\krprod\matr{Y}&=
\begin{pmatrix}
  \vec{x}_1\kprod\vec{y}_1, \vec{x}_2\kprod\vec{y}_2, \ldots, \vec{x}_m\kprod\vec{y}_m
\end{pmatrix}\\
&=\begin{pmatrix}
  x_{11}\vec{y}_1 &    \cdots & x_{1m}\vec{y}_m\\
  x_{21}\vec{y}_1 &    \cdots & x_{2m}\vec{y}_m\\
  \vdots            &    \ddots & \vdots              \\
  x_{n_11}\vec{y}_1 & \cdots & x_{n_1m}\vec{y}_m  \end{pmatrix}\,.
\end{align}
Notice that if $\matr{X}$ and $\matr{Y}$ are binary, so are $\matr{X}\kprod\matr{Y}$ and $\matr{X}\krprod\matr{Y}$.


%% file: data.tex
\section{Data Model}
\label{sec:dataModel}
Given an \RDF graph $T$, let $S(T)$, $P(T)$, and $O(T)$ denote the sets of distinct subjects, predicates and objects respectively. The number of distinct subjects is denoted by $\abs{S(T)}$, or shorthand $\abs{S}$. Correspondingly $\abs{P(T)} = \abs{P}$ and $\abs{O(T)} = \abs{O}$ denote the number of predicates and objects. The shorthand notation we use in unambiguous cases, for instance a maximal index occurring in a subscript, and if the cardinality is common to all \RDF graphs under discussion.
 
To represent \RDF data by a binary tensor, we enumerate all subjects, predicates, and objects to obtain mappings from the items in $S(T)$, $P(T)$, and $O(T)$ to indices. Let $s_i$ denote the $i$th subject with the corresponding index $i=1,\ldots,\abs{S}$ and respectively $p_j$ the $j$th predicate with $j=1,\ldots,\abs{P}$ and  $o_k$ the $k$th object with $k=1,\ldots,\abs{O}$.

With this mapping we can represent  any \RDF graph $T$ as a 3-way binary \byby{\abs{S}}{\abs{P}}{\abs{O}} tensor $\tens{T}$. An element $(i,j,k)$ of \tens{T} is $1$ if and only if the respective subject-predicate-object triple $(s_i,p_j,o_k)$ is present in the \RDF graph $T$.


%% file: operations.tex
\section{SPARQL Queries}
\label{sec:operations}
Simple \SPARQL queries consist of two parts, a \texttt{SELECT} clause and a \texttt{WHERE} clause. The \texttt{SELECT} clause identifies the variables to appear in the query result. The \texttt{WHERE} clause provides the basic graph pattern to match against the \RDF data. A basic graph pattern is a set of triple patterns that matches a subgraph of the \RDF data. Triple patterns are subject-predicate-object triples with the option that variables can be placed instead of each  specific subject, predicate, or object.

An example for a simple \SPARQL query that has a single triple pattern as basic graph pattern is \texttt{SELECT * WHERE \{?a $T$:p$_j$ $T$:o$_k$\}}. The keyword \texttt{SELECT} acts as a projection operator. It identifies the variables to appear in the query result. In this case, we ask for all variables that have been defined.
The triple pattern in the \texttt{WHERE} part of the query has a variable for the subject, \texttt{?a}, indicated by the prepended \texttt{?} and a fixed predicate and object from \RDF graph $T$, \texttt{p}$_j$ and \texttt{o$_k$}. It matches all \RDF triples of $T$ that have predicate $\texttt{p}_j$ and object  $\texttt{o}_k$.  

If the \RDF data is represented as a binary 3-way tensor \tens{T}, the triple pattern selects the fibre $\vec{t}_{:jk}$. That is a vector of all subjects with predicate $j$ and object $k$. This vector has a $1$ at positions $i$ that correspond to an \RDF triple $(s_i,p_j,o_k)$ present in the \RDF graph $T$.

A slice of \tens{T} would be selected if only one mode was fixed by the query. A query \texttt{SELECT * WHERE \{?a $T$:p$_j$ ?b\}} for instance resembles $\matr{T}_{:j:}$. If we change the projection to \texttt{SELECT ?a} we would get the same number of results but only the indices $\texttt{s}_i$ of the non-zero positions in the slice would appear in the final solution sequence and the indices $\texttt{o}_k$ remain hidden.

\subsection{Basic Graph Patterns -- Join}
\label{sec:operations:join}
The basic graph pattern where a set of triple patterns must match can be understood as a join operation. Consider for example a basic graph pattern consisting of two triple patterns, \texttt{\{?a $T$:p$_i$ ?b\} .\ \{?c $U$:p$_j$ ?b\}}. This queries for all \RDF triples where the object \texttt{?b} is linked to a subject by predicate \texttt{p}$_i$ of \RDF graph $T$ as well as by predicate \texttt{p}$_j$ of \RDF graph $U$,  where $i \in \{1,\ldots,\abs{P(T)}\}$ and $j \in \{1,\ldots,\abs{P(U)}\}$. 

Note that the braces in this example are only used to increase readability and could be omitted. For more complex queries however, curly braces define group graph pattern and hence the processing order.

In Boolean tensor algebra, the triple patterns from the example above resemble the slices $\matr{T}_{:i:}$ and $\matr{U}_{:j:}$ of \RDF tensors \tens{T} and \tens{U}. A join operation on equal objects is equivalent to the Khatri--Rao product of $\matr{T}_{:i:}$ and $\matr{U}_{:j:}$. However, in order to compute the Khatri--Rao product, the length of the columns of both slices must match and to obtain a meaningful result, the labels must be in the same order. 

Therefore, in case the objects of \tens{T} and \tens{U} do not map one-to-one, all-zero slices associated with the object labels from \tens{T} that have no correspondent object label in \tens{U} are appended to \tens{U} and vice versa, such that $\abs{O(T)} = \abs{O(U)}$. Furthermore the labels of the objects in \tens{T} and \tens{U} need to be in a common order.

The result to the basic graph pattern is a matrix of size \by{\abs{S(T)}\abs{S(U)}}{\abs{O}}, 
\begin{align}  \nonumber
\matr{T}_{:i:}\krprod\matr{U}_{:j:} &=
\begin{pmatrix}
  \vec{t}_{:i1}\kprod\vec{u}_{:j1}, \vec{t}_{:i2}\kprod\vec{u}_{:j2}, \ldots, \vec{t}_{:i\abs{O}}\kprod\vec{t}_{:u\abs{O}} 
\end{pmatrix}\\
  &= 
\left(\begin{matrix}
  t_{1i1}\vec{u}_{:j1} & _{1i2}\vec{u}_{:j2} & \cdots & t_{1i\abs{O}}\vec{u}_{:j\abs{O}}\\
  t_{2i1}\vec{u}_{:j1} & t_{2i2}\vec{u}_{:j2} & \cdots & t_{2i\abs{O}}\vec{u}_{:j\abs{O}}\\
  \vdots            &    \vdots        & \ddots & \vdots              \\
  t_{\abs{S}i1}\vec{u}_{:j1} & t_{\abs{S}i2}\vec{u}_{:j2} & \cdots & t_{\abs{S}i\abs{O}}\vec{u}_{:j\abs{O}}
\end{matrix}\right)\; .
\end{align}
This matrix has non-zeroes where objects have corresponding subjects when the predicate is \texttt{p}$_i$ as well as \texttt{p}$_j$. It can be regarded as $\abs{S(T)}$ blocks of \by{\abs{S(U)}}{\abs{O}} matrices stacked on top of each other. Each block then corresponds to a subject \texttt{?a} from $T$ and each row per block corresponds to a subject \texttt{?c} from $U$. This view enables a straight forward conversion from the Khatri--Rao product to the \RDF triples: Instead of numbering the row indices from $1$ to $\abs{S(T)}\abs{S(U)}$, we refer to each row by its block index and the row index within the block. These indices link to both subjects and the corresponding object is encoded by the column index.

There are more options to join triple patterns. Instead of binding the objects, we could bind on the subjects as in \texttt{\{?a $T$:p$_i$ ?b\} .\ \{?a $U$:p$_j$ ?c\}}, or on both as in \texttt{\{?a $T$:p$_i$ ?b\} .\ \{?a $U$:p$_j$ ?b\}}, or even on none as in \texttt{\{?a $T$:p$_i$ ?b\} .\ \{?c $U$:p$_j$ ?d\}}, or the subject of one pattern with the object of the other as in \texttt{\{?a $T$:p$_i$ ?b\} .\ \{?c $U$:p$_j$ ?a\}}, and so on. Also, it is possible to fix not the predicates but the subjects or objects, or even both of them or none. In the remainder of this section, we examine these options is detail and see how their Boolean tensor algebra counterparts look like. For this we assume the preprocessing step of matching the labels in the modes to be joined to be already done. 

\subsubsection{One Each Fixed, One Bound}

The first option for joining two triple patterns we examine is where one variable in each pattern is fixed, like the predicate in the example above, and one is bound. The bound variable is indicated by a common name in both triple patterns of the query. To be part of the result \RDF triples from both triple patterns must agree on the value of the bound variable. This means, we obtain all triple patterns that have a common label in the dimension of the bound variable.

Suppose we fix the predicates $T$:\texttt{p}$_i$ and $U$:\texttt{p}$_j$ in both the triple patterns, where $i\in \{1,\ldots,\abs{P(T)}\}$ and $j\in \{1,\ldots,\abs{P(U)}\}$. This leaves us with four options how to treat the subject and the object:
\begin{enumerate}
  \item \texttt{\{?a $T$:p$_i$ ?b\} .\ \{?a $U$:p$_j$ ?c\}}
  \item \texttt{\{?a $T$:p$_i$ ?b\} .\ \{?c $U$:p$_j$ ?b\}}
  \item \texttt{\{?a $T$:p$_i$ ?b\} .\ \{?c $U$:p$_j$ ?a\}}
  \item \texttt{\{?a $T$:p$_i$ ?b\} .\ \{?b $U$:p$_j$ ?c\}}
\end{enumerate}
As well as the predicate, we can fix either the subjects 
or the objects 
 or a combination of them. This gives us nine more options for each of the above combinations. 
So, in total there are 36 different ways to perform a join with one variable bound and one in each pattern fixed. When calculating the Khatri--Rao product, the cases we need to distinguish however are less. Fixing a variable corresponds to selecting a slice of the \RDF tensor. A slice is a matrix and hence we only care whether the join should be performed on the rows or columns of either matrix. The case where the columns of both matrices are joined, we have already seen in the example above. The case of joining the rows of both matrices can be accomplished by using the transpose of the matrices, and transposing the Khatri--Rao product. Thus, a join on the subjects while the predicates are fixed amounts to
\begin{align}\label{eq:joinSubj}
\left((\matr{T}_{:i:})^T\krprod(\matr{U}_{:j:})^T\right)^T &=
\begin{pmatrix}\nonumber
  \vec{t}_{1i:}\kprod\vec{u}_{1j:}, \vec{t}_{2i:}\kprod\vec{u}_{2j:}, \ldots, \vec{t}_{\abs{O}i:}\kprod\vec{u}_{\abs{O}j:} 
\end{pmatrix}^T
 \\ &= 
\left(\begin{matrix}
  t_{1i1}\vec{u}_{1j:} & t_{2i1}\vec{u}_{2j:} & \cdots & t_{\abs{O}i1}\vec{u}_{\abs{O}j:}\\
  t_{1i2}\vec{u}_{1j:} & t_{2i2}\vec{u}_{2j:} & \cdots & t_{\abs{O}i2}\vec{u}_{\abs{O}j:}\\
  \vdots            &    \vdots        & \ddots & \vdots              \\
  t_{1i\abs{S}}\vec{u}_{1j:} & t_{2i\abs{S}}\vec{u}_{2j:} & \cdots & t_{\abs{O}i\abs{S}}\vec{u}_{\abs{O}j:}
\end{matrix}\right)^T\; .
\end{align}

Analogously, to perform a join between the columns of the first matrix and the rows of the second, we compute $(\matr{T}_{:i:})\krprod(\matr{U}_{:j:})^T$ and obtain an \by{\abs{S(T)}\abs{O(U)}}{\abs{S(U)}} matrix where we interpret the rows as $\abs{S(T)}$ blocks of $\abs{O(U)}$ objects. Note that in order compute the Khatri--Rao product, the number of columns (i.e. the dimensions on which we join) must match and hence we require $\abs{S(U)} = \abs{O(T)}$ in this case. 

Likewise, to join between the rows of the first and the columns of the second, we need to evaluate $(\matr{T}_{:i:})^T\krprod\matr{U}_{:j:}$. This yields an \by{\abs{O(T)}\abs{S(U)}}{\abs{S(T)}} matrix where we interpret the rows as $\abs{O(T)}$ blocks of $\abs{S(U)}$ subjects.

\subsubsection{Two Each Fixed, One Bound}

As soon as two variables in a triple pattern are fixed, we select a vector from the corresponding \RDF tensor \tens{T}. The triples described by the pattern \texttt{\{s$_i$ $T$:p$_j$ ?a\}} for example amount to the non-zero entries in $\vec{t}_{ij:}$. A query like \texttt{\{$T$:s$_i$ $T$:p$_j$ ?a\} .\ \{$U$:s$_k$ $U$:p$_l$ ?a\}}, where $i\in \{1,\ldots,\abs{S(T)}\}$, $k\in \{1,\ldots,\abs{S(U)}\}$, $j\in \{1,\ldots,\abs{O(T)}\}$, and $l\in \{1,\ldots,\abs{O(U)}\}$, returns those triples where there is a $1$ in common positions in both vectors, i.e. $\vec{t}_{ij:}\wedge\vec{u}_{kl:}$. 

If we interpret the vectors $\vec{t}_{ij:}$ and $\vec{u}_{kl:}$ as two \by{\abs{O}}{1} matrices, we can use the Khatri--Rao product $\vec{t}_{ij:}\krprod\vec{u}_{kl:}$ to express the same operation.  Analogously, we can fix any two variables from each triple pattern and bind the remaining one. Furthermore, we can easily join also a triple pattern with two fixed variables with a triple pattern with one bound variable using the Khatri--Rao product. The query \texttt{\{$T$:s$_i$ $T$:p$_j$ ?a\} .\ \{?b $U$:p$_l$ ?a\}} for instance would yield a \by{1\cdot\abs{S(T)}}{\abs{O}} result.

\subsubsection{One Fixed, One Bound}

It is also an option to fix only one variable in one of the triple patterns. Now we examine the Boolean tensor algebra version of a query like \texttt{\{?a $T$:p$_j$ ?c\} .\ \{?d ?b ?c\}} with $j\in \{1,\ldots,\abs{P(T)}\}$. This query asks for all triples with predicate \texttt{$T$:p}$_j$ that have an object common to any other triple in the \RDF data. Thus, in case of an \RDF tensor representation, this amounts to a column-wise matrix--tensor multiplication. This type of multiplication can be achieved by computing the Khatri--Rao product between the matrix $\matr{T}_{:j:}$ representing the left triple pattern and each slice $\matr{U}_{:k:}$ of $\tens{U}$ along the predicates $k=1,\ldots,\abs{P(U)}$ together representing the right triple pattern. The outcome then is a 3-way tensor with $\abs{P(U)}$ slices, each of size \by{\abs{S(T)}\abs{S(U)}}{\abs{O}}.

Using the matricization of \tens{U} along the mode to be joined, in this case the object and hence the third mode $\matr{U}_{(3)}$, we can express the calculation described above in a concise form by
\begin{align}
\matr{T}_{:j:}\krprod\matr{U}_{(3)} &= 
\begin{pmatrix}
  \vec{t}_{:j1}\kprod\vec{u}_{(3):1}, \vec{t}_{:j2}\kprod\vec{u}_{(3):2}, \ldots, \vec{t}_{:j\abs{O}}\kprod\vec{u}_{(3):\abs{O}} 
\end{pmatrix} \nonumber\\
  &= 
\left(\begin{matrix}
  t_{1j1}\vec{u}_{(3):1} & t_{1j2}\vec{u}_{(3):2} & \cdots & t_{1j\abs{O}}\vec{u}_{(3):\abs{O}}\\
  t_{2j1}\vec{u}_{(3):1} & t_{2i2}\vec{u}_{(3):2} & \cdots & t_{2i\abs{O}}\vec{u}_{(3):\abs{O}}\\
  \vdots            &    \vdots        & \ddots & \vdots              \\
  t_{\abs{S}j1}\vec{u}_{(3):1} & t_{\abs{S}j2}\vec{u}_{(3):2} & \cdots & t_{\abs{S}j\abs{O}}\vec{u}_{(3):\abs{O}}
\end{matrix}\right)\; ,
\end{align}
where $\matr{U}_{(3)}$ is a matrix of size \by{\abs{P(U)}\abs{S(U)}}{\abs{O}}. The result is thus of size \by{\abs{S(T)}\abs{P(U)}\abs{S(U)}}{\abs{O}} and in order to translate back to an \RDF graph, the first dimension is regarded as $\abs{S(T)}$ blocks each with $\abs{P(U)}$ blocks of $\abs{S(U)}$ items. This view enables to address every field in the result matrix with a tuple \texttt{(s$_i$ p$_j$ s$_k$ o$_l$)}, where $i\in \{1,\ldots,\abs{S(T)}\}$, $j\in \{1,\ldots,\abs{P(U)}\}$, $k\in \{1,\ldots,\abs{S(U)}\}$, and $l\in \{1,\ldots,\abs{O}\}$. The positions of the non-zeroes addressed in this way answer the \RDF query. 

Of course, if a similar query is posed with the subjects bound instead of the objects, we need to apply the transpose before computing the Khatri--Rao product, and as well transpose the outcome. (This is similar to what has been discussed in case of one variable fixed in each triple.)

\subsubsection{Some Fixed, None Bound}

We now examine joins between triple patterns where none of the variables is bound. Technically, these are not any more joins. The simplest such case that is worth looking at is where in each of the triple patterns two variables are fixed. There is also the case of all three variables fixed, but this is no more than to compare whether two \RDF triples are equal.

A query with two variables fixed in each triple pattern for example is \texttt{\{$T$:s$_i$ $T$:p$_j$ ?a\} .\ \{$U$:s$_k$ $U$:p$_l$ ?b\}}\,, where $i\in \{1,\ldots,\abs{S(T)}\}$, $j\in \{1,\ldots,\abs{P(T)}\}$. $k \in \{1,\ldots,\abs{S(U)}\}$, and $l\in \{1,\ldots,\abs{P(U)}\}$. The expected result from that query is a list of all combinations of the triples from the left pattern and the triples from the right pattern. As triples in each of the patterns differ only in their objects, this means we are looking for all possible combinations of objects matching the left triple pattern and objects matching the right triple pattern. These we obtain by taking the outer product of $\vec{t}_{ij:}$ and $\vec{u}_{kl:}$, that is $\vec{t}_{ij:}(\vec{u}_{kl:})^T$.

Similarly, for a join query with one variable fixed in each triple pattern such as \texttt{\{?a $T$:p$_j$ ?b\} .\ \{?c $U$:p$_l$ ?d\}} the result is all combinations of subject-object pairs from the left triple pattern with subject-object pairs from the right triple pattern. In terms of Boolean algebra such operations can be expressed using the Kronecker product, a  generalization of the outer product,
\begin{align}
\matr{T}_{:j:}\kprod\matr{U}_{:l:} = 
\left(\begin{matrix}
  t_{1j1}\matr{U}_{:l:} & t_{1j2}\matr{U}_{:l:} & \cdots & t_{1j\abs{O}}\matr{U}_{:l:}\\
  t_{2j1}\matr{U}_{:l:} & t_{2j2}\matr{U}_{:l:} & \cdots & t_{2j\abs{O}}\matr{U}_{:l:}\\
  \vdots            &    \vdots        & \ddots & \vdots              \\
  t_{\abs{S}j1}\matr{U}_{:l:} & t_{\abs{S}j2}\matr{U}_{:l:} & \cdots & t_{\abs{S}j\abs{O}}\matr{U}_{:l:} 
\end{matrix}\right).
\end{align}

An even more general notion of the outer product is needed to express a join between two triple patterns when no variable is fixed or bound, as in \texttt{\{?a ?b ?c\} . \{?d ?e ?f\}}. This can be accomplished multiplying two tensors, $\tens{T}\kprod\,\tens{U}$, as discussed in~\cite{bader2006algorithm}. This operation, although it is possible appears to have little practical relevance.  

\subsubsection{None Fixed, Some Bound}

The other extreme scenario is to have no fixed variables but varying amounts of bound variables in the triple patterns to join. The last case discussed in the previous section already is an example of such a join. The other end however is much simpler to start with: Consider a join query where all variables are bound. That would be \texttt{\{?a ?b ?c\} .\ \{?a ?b ?c\}}, which is asking for all \RDF triples that occur in the \RDF graph $T$ and that match all \RDF triples that occur in the \RDF graph $U$. So, this query just yields all triples that occur in both $T$ and $U$. In Boolean tensor algebra, this amounts to the position-wise \textsc{or} between the corresponding tensors \tens{T} and \tens{U} of common dimension \byby{\abs{S}}{\abs{P}}{\abs{O}} (possibly achieved by a preprocessing step),
\begin{align}
\bigvee_{i=1}^{\abs{S}}\bigvee_{j=1}^{\abs{P}}\bigvee_{k=1}^{\abs{O}} \vec{t}_{ijk}\kprod\vec{u}_{ijk}\,,
\end{align}
 where $i=1,\ldots,\abs{S}$, $j=1,\ldots,\abs{P}$, and $k=1,\ldots,\abs{O}$.

Similarly, if two variables are bound as in \texttt{\{?a ?b ?c\} .\ \{?a ?e ?c\}}), we receive only those triples where the subjects and objects match each other. This means, we take the outer product of each tube $\vec{t}_{i:k}$ from \tens{T} with the corresponding tube $\vec{u}_{i:k}$ from \tens{U},
\begin{align}
\bigvee_{i=1}^{\abs{S}}\bigvee_{k=1}^{\abs{O}} \vec{t}_{i:k}\kprod\vec{u}_{i:k}\,,
\end{align}
 where $i=1,\ldots,\abs{S}$, $k=1,\ldots,\abs{O}$. 

In case one variable bound as in \texttt{\{?a ?b ?c\} .\ \{?a ?e ?f\}} we receive only those triples where the subjects match each other. This means that, we take the Kronecker product of each slice $\matr{T}_{i::}$ from \tens{T} with the corresponding slice $\matr{U}_{i::}$ from \tens{U}. For $i=1,\ldots,\abs{S}$, we get
\begin{align}
\bigvee_{i=1}^{\abs{S}} \matr{T}_{i::}\kprod\matr{U}_{i::}\,,
\end{align}
a matrix is of size \by{\abs{S(T)}\abs{S(U)}}{\abs{O(T)}\abs{O(U)}}.

\subsection{
Optional Graph Patterns -- Left Outer Join}
\label{sec:operations:optional}
Apart from the basic graph patterns, there are various other ways to combine triple patterns: group graph patterns, optional graph patterns, alternative graph patterns, and patterns on named graphs. Group graph patterns define the evaluation hierarchy which directly translates to associativity of Boolean tensor operations. Alternative graph patterns refer to the union of triple patterns which easily translates to \textsc{or}ing the respective tensor slices. Patterns on named graphs refer to the possibility to join triple patterns from different \RDF graphs. 

The combination of triple patterns discussed now are the optional graph patterns. These patterns resemble left outer joins. This means, all triples from the left triple pattern have to appear in the final result. If they do not match any triple from the right pattern, they are listed paired with a blank item. Hence, we use all the triples that match, like in the normal join, and some more. To accomplish a left outer join operation using Boolean operations, we can thus use the join operation discussed above, but additionally we need to handle blank items.

Consider the case where we join two triple patterns with fixed predicates and bound subjects (Equation \ref{eq:joinSubj}). In the Khatri--Rao product, we get a $1$ at positions where the left and right triples match with their subjects. Additionally now we need to cover triples with subjects that occur in the left triple pattern but not in the right. To do that we append a column to $\matr{U}_{:j:}$, the slice corresponding to the right triple pattern. That column has a $1$ in rows where $\matr{U}_{:j:}$ is all-zero and $\matr{T}_{:i:}$ (the slice corresponding to the left triple pattern) has at least one non-zero. To evaluate \texttt{\{?a $T$:p$_i$ ?b\} OPTIONAL \{?a $U$:p$_j$ ?c\}}, we first append the extra column to $\matr{U}_{:j:}$ and obtain
\begin{align}
\matr{U'}_{:j:} = \bigl[\matr{U}_{:j:}, \vec{k}\bigr] = \left[\matr{U}_{:j:}, \bigvee_{r=1}^{\abs{O(U)}}\vec{u}_{:ir} \wedge \neg\bigvee_{s=1}^{\abs{O(U)}}\vec{u}_{:js}\right]\,.
\end{align}
Then, like for the join on subjects, we compute the transposed Khatri--Rao product of the transposes of $\matr{T}_{:i:}$ and $\matr{U'}_{:j:}$, and receive an \by{\abs{S}}{\abs{O(T)}(\abs{O(U)}+1)} matrix as the result,
\begin{align}
&\left((\matr{T}_{:i:})^T\krprod(\matr{U'}_{:j:})^T\right)^T =
\begin{pmatrix}
  \vec{t}_{1i:}\kprod\vec{u'}_{1j:}, \vec{t}_{2i:}\kprod\vec{u'}_{2j:}, \ldots, \vec{t}_{\abs{O}i:}\kprod\vec{u'}_{(\abs{O}+1)j:} 
\end{pmatrix}^T\,.
\end{align}

The additional column we introduce actually stands for ``no value'' and should be treated accordingly in any follow up calculation on the obtained result.

\subsection{Unique Results -- Select Distinct}
\label{sec:operations:distinct}
Until now, we discussed different \texttt{WHERE} clauses. This section focuses on the keyword \texttt{DISTINCT}, a modifier of the \texttt{SELECT} clause. This keyword states that the solutions in the result sequence of the query must be unique, hence no duplicate solutions can occur. The na\"ive way to approach this type of query is to first compute the solution from the \texttt{WHERE} clause and then purge the duplicates. Of course, this can also be accomplished similarly treating the \RDF data as a binary tensor: Compute the result for the \texttt{WHERE} clause using the operations stated above, then \textsc{or} along the dimensions not asked for in the \texttt{SELECT} clause. But the binary tensor representation offers also a more straight forward approach for such queries: For a join query accompanied by \texttt{DISTINCT}, instead of computing the Khatri--Rao product and then \textsc{or}, the result can be obtained immediately.

The first case we examine is a join query where \texttt{SELECT DISTINCT} chooses the bound variable, such as \texttt{SELECT DISTINCT ?b WHERE \{?a $T$:p$_i$ ?b\} .\ \{?c $U$:p$_j$ ?b\}}\,.  Suppose we use the Khatri--Rao product to calculate the result of the \texttt{WHERE} clause. Then, the distinct objects \texttt{?b} are those corresponding to the non-zero columns,
\begin{align}
\bigvee_{k=1}^{\abs{S(T)}} (\matr{T}_{:i:} \krprod \matr{U}_{:j:})_{k:}\,.
\end{align}
In fact, the objects that constitute the result are those which occur at least once in both \matr{T}$_{:i:}$ and \matr{U}$_{:j:}$. Hence, this calculation simplifies to evaluating whether both the columns of \matr{T}$_{:i:}$ and \matr{U}$_{:j:}$ are non-zero,
\begin{align}
\left(\bigvee_{k=1}^{\abs{S(T)}} \vec t_{ki:}\right) \wedge \left(\bigvee_{k=1}^{\abs{S(U)}} \vec u_{kj:}\right)\,.
\end{align}
The resulting vector of length $\abs{O}$ has a $1$ at positions that refer to objects which are part of the result and zeroes elsewhere.

The next case is to uniquely select pairs of one bound and one free variable, as for instance in \texttt{SELECT DISTINCT ?a ?b WHERE \{?a $T$:p$_i$ ?b\} .\ \{?c $U$:p$_j$ ?b\}}\,. To evaluate this query, we first calculate the \by{\abs{S(T)}\abs{S(U)}}{\abs{O}} Khatri--Rao product as discussed in Section \ref{sec:operations:join}. Note that we treat the rows of the resulting matrix as $\abs{S(T)}$ blocks, each of $\abs{S(U)}$ columns. The next step then is to get a matrix of size \by{\abs{S(T)}}{\abs{O}} that has a $1$ at position $(m, n)$ if in the $m$-th block of the Khatri--Rao product the $n$-th column has at least one non-zero (with $m\in \{1,\ldots, \abs{S(T)}\}$ and $n \in \{1,\ldots,\abs{O}\}$). More straightforward, the same answer is obtained by taking those columns of \matr{T}$_{:i:}$ where the corresponding columns of \matr{U}$_{:j:}$ have at least one non-zero,
\begin{align}
\matr T_{:i:} \circ \left(\bigvee_{k=1}^{\abs{S(U)}} \vec u_{kj:}\right)\,.
\end{align}

The final case we examine is to uniquely select pairs of the unbound variables, as in the query \texttt{SELECT DISTINCT ?a ?c WHERE \{?a $T$:p$_i$ ?b\} .\ \{?c $U$:p$_j$ ?b\}}\,, where $i \in \{1,\ldots,\abs{P(T)}\}$ and $j \in \{1,\ldots,\abs{P(U)}\}$. We show how the na\"ive way to first compute the Khatri--Rao product followed by \textsc{or}-ing along the columns corresponds to the Boolean matrix product,
\begin{align}
\matr{T}_{:i:}\circ\matr{U}_{:j:}\,.
\end{align}
The query asks for all unique \texttt{?a}--\texttt{?c} pairs that match on \texttt{?b}. As described in Section \ref{sec:operations:join}, to get all \texttt{?a}--\texttt{?c} pairs that match on \texttt{?b} we compute the Khatri--Rao product of the respective slices of \tens{T} and \tens{U}. To retrieve only unique results in terms of subject--subject combinations, we compute \textsc{or} along the columns of the Khatri--Rao product (corresponding to the objects). This is 
\begin{align}\label{eq:oper:dist:kr}
 \bigvee_{k=1}^{\abs{O}}\left(\matr{T}_{:i:}\krprod\matr{U}_{:j:}\right)_{:k}  &=
\bigvee_{k=1}^{\abs{O}}\begin{pmatrix}
  \vec{t}_{:i1}\kprod\vec{u}_{:j1}, \vec{t}_{:i2}\kprod\vec{u}_{:j2}, \ldots, \vec{t}_{:i\abs{O}}\kprod\vec{u}_{:j\abs{O}} 
\end{pmatrix}_{:k}\\
  &= 
\left(\begin{matrix}
\bigvee_{k=1}^{\abs{O}}& t_{1ik}u_{1jk} \\
\bigvee_{k=1}^{\abs{O}}& t_{1ik}u_{2jk} \\
&\vdots\\
\bigvee_{k=1}^{\abs{O}}& t_{1ik}u_{\abs{S}jk} \\
\bigvee_{k=1}^{\abs{O}}& t_{2ik}u_{1jk} \\
\bigvee_{k=1}^{\abs{O}}& t_{2ik}u_{2jk} \\
&\vdots\\
\bigvee_{k=1}^{\abs{O}}& t_{\abs{S}ik}u_{\abs{S}jk} \\
\end{matrix}\right) = \left(\matr{T}_{:i:}\circ\matr{U}_{:j:}\right)_{(1)} \label{eq:oper:dist:bool} \;,
\end{align}
the vectorization of the Boolean matrix product $\matr{T}_{:i:}\circ\matr{U}_{:j:}$ along the columns. This matches the well-known fact that join-distinct in standard relational databases can be computed using the Boolean matrix product. In particular, we can as well apply the techniques for fast computation of these joins to this case (cf. Section~\ref{sec:join-card:bool}).

One conceptual difference between first computing the initial solution sequence from the \texttt{WHERE} clause and then applying \texttt{DISTINCT} over computing the result in one step is that the specified execution order cannot be obeyed. \SPARQL defines that \texttt{ORDER BY} statements must be applied on the initial solution sequence before the projection. This means, after processing everything from the \texttt{WHERE} clause but before anything from the \texttt{SELECT} statement. However, by using binary tensors instead of graphs to represent the \RDF data, and in particular by using lists of labels instead of sets, we already introduce an ordering even before processing the \texttt{WHERE} clause. If this ordering obeys the \texttt{ORDER BY} statement, we can also represent the output in the appropriate order without the intermediate step.

Note that \SPARQL also defines the keyword \texttt{REDUCED} that modifies the \texttt{SELECT} clause such that it outputs at most all results that \texttt{SELECT} without modifiers would yield and at least those results that \texttt{SELECT DISTINCT} yields. Hence, there is no new operation to define for that modifier, as we already discussed two options to provide a valid answer to a \texttt{SELECT REDUCED} query.

\subsection{Matching Alternatives -- Union}
\label{sec:operations:union}
The keyword \texttt{UNION} may be placed between two graph patterns in a \texttt{WHERE} clause. The result of a union query is a concatenation of the solution sequences of the graph patterns. We can interpret this as an \textsc{or} between the results of both graph patterns. Note however that no matching takes place in this type of query. The initial solution sequence comprises all results from both graph patterns, Each result is equipped with all the variables defined that occur in either of the graph patterns, possibly left blank. 

Consider \texttt{SELECT ?b WHERE \{?a $T$:p$_i$ ?b\} UNION \{?c $U$:p$_j$ ?b\}}. The initial solution sequence comprises all \texttt{?a}--\texttt{?b} pairs with an additional blank \texttt{?c}, as well as all  \texttt{?c}--\texttt{?b} pairs with an additional blank \texttt{?a}. This is the same as concatenating \matr{U}$_{:i:}$ and \matr{T}$_{:j:}$ along the columns. The final result would be all objects from $T$ with predicate \texttt{p}$_i$ and all objects from $U$ with predicate \texttt{p}$_j$. In our terminology, it would be a vector of length $\abs{O(T)}+\abs{O(U)}$ which has a $1$ at positions that either refer to an object of $T$ or $U$ that occurs in the result. Similarly, in  a query like \texttt{SELECT ?a ?c WHERE \{?a $T$:p$_i$ ?b\} UNION \{?c $U$:p$_j$ ?d\}} \matr{U}$_{:i:}$ and \matr{T}$_{:j:}$ would be concatenated along the diagonal, as no variable is bound. The final solution sequence would comprise all subjects from $T$ and all subjects from $U$, bound to separate variable names, each paired with a blank.

\subsection{Further Operations}
\label{sec:operations:others}
While join and union operations directly act on the data structure, \SPARQL also defines other types of operations that act on the labels of the data. For instance, it is possible to order the solution sequence using a comparator on the labels of one mode. Also, one can use filters to restrict the solutions to those for which the filter expression yields \textsc{true}. Filters can for example be comparators to restrict numeric values, or regular expressions to restrict the values of strings. 

The \texttt{ASK} statement can be used in place of \texttt{SELECT}. The result of such a query is \textsc{true} if the solution sequence obtained from the \texttt{WHERE} part is non-empty. Hence, if we calculate the result on Boolean tensors, we emit \textsc{true} as soon as the first non-zero appears in the result.

\SPARQL provides the keyword \texttt{CONSTRUCT} that allows generate a new \RDF graph from the result of a query. As any \RDF graph can be expressed as a Boolean tensor, an \RDF tensor can as well be constructed by such a query.


%% file: cardinality.tex
\section{Cardinality and Computation of Joins}
\label{sec:cardinality}

An important problem in query processing is to estimate the cardinality of join operations. This problem has attracted a significant amount of research in traditional relational databases. For \SPARQL joins, Neumann and Moerkotte~\cite{neumann11characteristic} proposed so called \emph{characteristic sets} to estimate the cardinalities of \SPARQL joins (specifically, the type of joins they called \emph{star joins}). In this section we study how the size of  join operations can be computed (or estimated) given our framework.

\subsection{Khatri--Rao Products}
\label{sec:join-card:kr}

Let us first assume we have stored the marginal sums along each mode of the \byby{\abs{S}}{\abs{P}}{\abs{O}} data tensor $\tens{T}$. That is, we have three matrices, $\matr{P}$ (\by{\abs{P}}{\abs{O}}), $\matr{Q}$ (\by{\abs{S}}{\abs{O}}), and $\matr{R}$ (\by{\abs{S}}{\abs{P}}), for the column, row, and tube marginal sums, respectively. The element $(i,j)$ of $\matr{Q}$, for example, would be computed as $r_{ij} = \sum_{k=1}^{\abs{P}} t_{ijk}$.

The number of triples returned by a join, with no projection or \texttt{DISTINCT} keyword, is the number of non-zeroes in the result. When a join can be expressed as a Khatri--Rao product between two matrices $\matr{A}$ and $\matr{B}$ (as is the case with most joins considered in Sections~\ref{sec:operations:join} and~\ref{sec:operations:optional}), the number of non-zeroes in $\matr{A}\krprod\matr{B}$ can be determined exactly using the column marginal sums of $\matr{A}$ and $\matr{B}$. Specifically, if $\matr{A}$ and $\matr{B}$ both have $n$ columns, let $\vec{\sigma}^{\matr{A}} = (\sigma^{\matr{A}}_i)_{i=1}^n$ and $\vec{\sigma}^{\matr{A}} = (\sigma^{\matr{B}}_i)_{i=1}^n$ be row vectors that contain the column marginals of $\matr{A}$ and $\matr{B}$, respectively (e.g.\ $\sigma^{\matr{A}}_i = \sum_{j} a_{ji}$). 

\begin{proposition}
\label{prop:card:kr}
Let $\vec{\sigma}^{\matr{A}}$ and $\vec{\sigma}^{\matr{B}}$ be as above. The number of non-zeroes in $\matr{A}\krprod\matr{B}$ is
\begin{equation}
  \label{eq:card:kr}
  \abs{\matr{A}\krprod\matr{B}} = \sum_{i=1}^n \sigma^{\matr{A}}_i \sigma^{\matr{B}}_i = \vec{\sigma}^{\matr{A}}(\vec{\sigma}^{\matr{B}})^T\; .
\end{equation}
  \end{proposition}

The column marginal vectors $\vec{\sigma}$ are naturally just appropriate rows or columns of the tensor marginal sum matrices $\matr{P}$, $\matr{Q}$, or $\matr{R}$. If they are stored in a sparse format, the size of the join can be computed exactly in time $\Theta(\alpha + \beta)$, where $\alpha$ and $\beta$ are the number of non-empty columns of $\matr{A}$ and $\matr{B}$, respectively.

We can also obtain an upper bound for the size of the join in constant time if in addition we store the $l_2$-norms for each row and column of the marginal sum matrices (that is, $\norm{\vec{\sigma}}$ for every possible $\vec{\sigma}$):

\begin{proposition}
  \label{prop:card:cosine}
  Let $\vec{\sigma}^{\matr{A}}$ and $\vec{\sigma}^{\matr{B}}$ be as above. Then
  \begin{equation}
    \label{eq:card:cosine}
    \abs{\matr{A}\krprod\matr{B}} \leq \norm{\vec{\sigma}^{\matr{A}}}\norm{\vec{\sigma}^{\matr{B}}}\; .
  \end{equation}
\end{proposition}

\begin{proof}
  Noticing that $\vec{\sigma}^{\matr{A}}(\vec{\sigma}^{\matr{B}})^T = \norm{\vec{\sigma}^{\matr{A}}}\norm{\vec{\sigma}^{\matr{B}}}\cos\theta$ together with~\eqref{eq:card:kr} gives the result as $\cos\theta\leq 1$.
\end{proof}

As of now, writing the join as a Khatri--Rao product does not seem to bring significant benefits on computing the full join, though. The best worst-case bound is $O(\abs{\matr{A}}\abs{\matr{B}})$ from the straight-forward evaluation of the algorithm. 

As mentioned in Section~\ref{sec:operations:distinct}, those \texttt{SELECT DISTINCT} queries that choose the bound variable can be computed by selecting the non-zero columns of the Khatri--Rao product. Consequently, the cardinality of the result is $\vec{\sigma}^{\matr{A}}(\vec{\sigma}^{\matr{B}})^T$ (taking vectors $\vec{\sigma}$ as row vectors) and the whole query (not just the cardinality) can be computed in time $\Theta(\alpha + \beta)$.

All of the above computations require an access to the marginal sums. We argue that storing (and updating) this information is feasible for the data base system. In principle, the matrices can be very large, but notice that every $(s,p,o)$ triple of the data has effect to exactly one element in each of the three matrices. Hence, the total number of non-zeroes in the marginal matrices cannot exceed $3\abs{\tens{T}}$, and in practice it would be expected to stay much smaller. 

\todo{Estimate the size/result of distinct queries e.g. via min-wise hashing}

\subsection{Boolean Products}
\label{sec:join-card:bool}

When the \texttt{SELECT DISTINCT} query asks for a pair of variables, the evaluation does not (have to) involve the Khatri--Rao product, but rather the Boolean matrix product (Section~\ref{sec:operations:distinct}). Here, estimating the size of the result becomes harder. Take, for example~\eqref{eq:oper:dist:kr}: computing the result involves taking an \textsc{or} over the columns of a Khatri--Rao product, and hence, even if we know the cardinalities of each column, we can only give very coarse estimations on the final cardinality:
\begin{proposition}
\label{prop:card:union-bound}
Let $\matr{A}$ and $\matr{B}$ be two binary matrices with $n$ columns and let $\vec{\sigma}^{\matr{A}}$ and $\vec{\sigma}^{\matr{B}}$ be their column marginal sum vectors. Then,
\begin{equation}
  \label{eq:card:union-bound}
  \max_{i=1}^n\{\sigma^{\matr{A}}_i \sigma^{\matr{B}}_i\} \leq \abs{\bigvee_{i=1}^n(\matr{A}\krprod\matr{B})_{:i}} \leq \sum_{i=1}^n \sigma^{\matr{A}}_i \sigma^{\matr{B}}_i \; .
\end{equation}  
\end{proposition}

As a Khatri--Rao product, $\matr{A}\krprod\matr{B}$, contains a specific structure that could help estimating the cardinality better. Another approach is to estimate the cardinality from the Boolean product formulation~\eqref{eq:oper:dist:bool}. We will first sketch some estimates on the expectation that are simple and fast to compute, and could therefore be of interest to practitioners. 

\begin{proposition}
\label{prop:card:uar:matrix}
Let $\matr{A}$ and $\matr{B}$ be \by{m}{k} and \by{k}{n} binary matrices whose non-zeroes are uniformly distributed, and let $p^{\matr{A}}$ and $p^{\matr{B}}$ be the densities of $\matr{A}$ and $\matr{B}$, respectively. Then
\begin{equation}
  \label{eq:card:uar:matrix}
  \E[\abs{\matr{A}\bprod\matr{B}}] = 1 - \left(1 - p^{\matr{A}}p^{\matr{B}}\right)^k \; .
\end{equation}
\end{proposition}

We can improve our estimation of the expectation if we notice that the Boolean product is equivalent to element-wise \textsc{or}s of $k$ rank-1 binary matrices. 

\begin{proposition}
  \label{prop:card:uar:rank1}
  Let $\matr{A}$ and $\matr{B}$ be \by{m}{k} and \by{k}{n} binary matrices and let $\vec{p}^{\matr{A}} = (p ^{\matr{A}}_i)_{i=1}^k$ and $\vec{p}^{\matr{B}} = (p^{\matr{B}}_i)_{i=1}^k$ be the column densities of $\matr{A}$ and row densities of $\matr{B}$, respectively.  Assuming the non-zeroes in the rank-1 matrices $\vec{a}_{:i}\vec{b}_{i:}$ are distributed uniformly at random for all $i=1,\ldots,k$, we have that
  \begin{equation}
    \label{eq:card:uar:rank1}
    \E[\abs{\matr{A}\bprod\matr{B}}] = 1 - \prod_{i=1}^k p^{\matr{A}}_i p^{\matr{B}}_i \; .
  \end{equation}
\end{proposition}

Notice that if we take the union bound of the densities of the rank-1 matrices to obtain the upper bound for the number of non-zeroes in $\matr{A}\bprod\matr{B}$, we obtain the same result as in Proposition~\ref{prop:card:union-bound}.

The above methods, while being straight forward, require only the marginal sums, which makes them relatively fast to compute. If we can access the whole matrices, we can do much better estimations. In particular, we can use the result of Amossen et al.~\cite{amossen10better}:

\begin{proposition}[$\!$\cite{amossen10better}]
  \label{prop:card:amossen}
  There exists an algorithm that obtains an $1+\varepsilon$ approximation of $\abs{\matr{A}\bprod\matr{B}}$ in expected time $O(\abs{\matr{A}} + \abs{\matr{B}})$ for any $\varepsilon > 4/\sqrt[4]{\abs{\matr{A}} + \abs{\matr{B}}}$.
\end{proposition}

Computing the Boolean product can also be done faster than the standard matrix product: for example, Amossen and Pagh~\cite{amossen09faster} proposed an algorithm that runs in time $\tilde{O}(s^{2/3}z^{2/3} + s^{0.862}z^{0.408})$, where $s = \abs{\matr{A}} + \abs{\matr{B}}$ is the number of non-zeroes in the input, $z=\abs{\matr{A}\bprod\matr{B}}$ is the number of non-zeroes in the output, and $\tilde{O}(\cdot)$ hides the polylogarithmic factors. For sparse input matrices, the bound is generally very good, but if the input matrices are quite dense (and the output moderately so), the overall time complexity exceeds that of the standard matrix multiplication, i.e. $\tilde{O}(m^{\omega})$ for \by{m}{m} matrices. To address this, Lingas~\cite{lingas2011fast} proposed a randomized algorithm that runs in time $\tilde{O}(m^2s^{\omega/2 - 1})$, where the factor matrices are \by{m}{m} and $s$ is as above. (See~\cite{lingas2011fast} for an extension to rectangular matrices.)


%% file: decompositions.tex
\section{Tensor Decompositions}
\label{sec:decompositions}

Similarly to matrices, decomposition is a natural operation to tensors. A tensor decomposition reveals regularities of the decomposed tensor, and these regularities can sometimes speed up the computations. Furthermore, finding the regularities is a natural data analysis task. In this section, we will first give the definition of (Boolean) tensor CP decomposition and (Boolean) tensor rank. We will then study the properties of the decompositions, especially the sparsity. Unlike for the normal decomposition, with Boolean decompositions, we can prove upper bounds on the density of the factor matrices, showing that storing the data in a decomposed format is a valid option for saving space.  

\subsection{Definitions}
\label{sec:decompositions:definition}

The so-called \emph{Boolean tensor CP decomposition}\!\footnote{The name is short for two names
  given to the same decomposition: CANDECOMP~\cite{carroll70analysis}
  and PARAFAC~\cite{harshman70foundations}.} is defined as follows:

\begin{definition}[$\!$\cite{miettinen11boolean}]
 \label{def:BCP}
  Given an \byby{n}{m}{l} binary tensor $\tens{T}$ and an integer $r$, its \emph{rank-$k$ Boolean CP decomposition} consists of three binary matrices $\matr{A}$ (\by{n}{r}), $\matr{B}$ (\by{m}{r}), and $\matr{C}$ (\by{l}{r}) such that 
  \begin{equation}
    \label{eq:BCP}
    \tens{T}=\bigvee_{i=1}^r\vec{a}_{:i}\outprod\vec{b}_{:i}\outprod\vec{c}_{:i}\; .
  \end{equation}
\end{definition}

The Boolean CP decomposition is closely connected to the concept of \emph{Boolean tensor rank}. 

\begin{definition}
  \label{def:Brank}
  A 3-way binary tensor $\tens{T}$ is \emph{rank-1 tensor} if $\tens{T}$ is an outer product of three binary vectors, that is
  \begin{equation}
    \label{eq:Brank1}
    \tens{T} = \vec{a} \outprod \vec{b} \outprod \vec{c}\; .
  \end{equation}
  The \emph{Boolean rank} of a 3-way binary tensor $\tens{T}$, denoted $\rankB(\tens{T})$, is the least integer $r$ such that there exist $r$ rank-$1$ binary tensors with
  \begin{equation}
    \label{eq:Brank}
    \tens{T} = \bigvee_{i=1}^r \vec{a}_{:i} \outprod \vec{b}_{:i} \outprod \vec{c}_{:i}\; . 
  \end{equation}
\end{definition}

Notice that the $i$th columns of the \emph{factor matrices}
$\matr{A}$, $\matr{B}$, and $\matr{C}$ of the CP decomposition define a rank-$1$ tensor
$\vec{a}_{:i}\outprod\vec{b}_{:i}\outprod\vec{c}_{:i}$. In other words, the CP
decomposition expresses the given tensor as a Boolean sum of $r$ rank-$1$ tensors.

The Boolean CP decomposition can also be expressed in terms of Boolean and Khatri--Rao matrix products using matricization: three binary matrices $\matr{A}$, $\matr{B}$, and $\matr{C}$ form a Boolean CP decomposition of $\tens{T}$ if and only if~\cite{miettinen11boolean}

\begin{align}
  \label{eq:BCP:unfold}
    \matr{T}_{(1)} &= \matr{A}\bprod(\matr{C}\krprod\matr{B})^T\; , \\
    \matr{T}_{(2)} &= \matr{B}\bprod(\matr{C}\krprod\matr{A})^T\; , \\
\intertext{and}
    \matr{T}_{(3)} &= \matr{C}\bprod(\matr{B}\krprod\matr{A})^T \; .
\end{align}

Intuitively, then, we can think of the operation that turns the three factor matrices into a tensor as an operation that takes a join of two slices followed by distinct join of the result and another slice.

\subsection{Properties of the CP Decomposition}
\label{sec:decompositions:sparsity}

Unfortunately, computing the rank or CP decomposition of a tensor is \NP-hard, both under the normal~\cite{hastad90tensor} and Boolean~\cite{miettinen11boolean} algebras. The rank is also not bounded by the smallest dimension, unlike with matrices. That is, there exist \byby{n}{m}{l} binary tensors $\tens{T}$ such that $\rankB(\tens{T}) > \min\{n, m, l\}$~\cite{miettinen11boolean}. However, we have an upper bound~\cite{miettinen11boolean},
\begin{equation}
  \label{eq:brank:bound}
  \rankB(\tens{T}) \leq \min\{nm, nl, ml\}\; .
\end{equation}

Given the above results, we cannot hope for an efficient algorithm finding the smallest CP decomposition of data. Yet, we can always find \emph{some} CP decomposition. The most na\"ive way is to unfold the data along the longest mode (so that the unfolded matrix has $\min\{nm, nl, ml\}$ columns), set this as one factor matrix, and then construct the other two factor matrices in such a way that their Khatri--Rao product is the identity matrix (how that is done is explained in~\cite{miettinen11boolean}). Henceforth we will assume that our data tensor is represented in the factorized format. 

A common motivation for storing data in factorized formats is that they can save space compared to storing the full data, essentially by a more efficient representation of regularities in the data. Yet, there usually are no studies on how much space the decomposition could save, or take. Boolean decompositions, however, do allow us to bound the density (or sparsity) of the factors with respect to the density of the data. First we repeat the result on the absolute number of non-zeroes in the factor matrices, from~\cite{miettinen11boolean}.

\begin{proposition}
  \label{prop:sparsity:abs}
  Let $\tens{T}$ be binary a tensor that has $\rankB(\tens{T}) = r$. Then $\tens{T}$ has a rank-$r$ Boolean CP decomposition to $\matr{A}$, $\matr{B}$, and $\matr{C}$ such that
  \begin{equation}
    \label{eq:sparsity:abs}
    \abs{\matr{A}} + \abs{\matr{B}} +\abs{\matr{C}} \leq 3\abs{\tens{T}}\; .
  \end{equation}
\end{proposition}

Inequality~\eqref{eq:sparsity:abs} is tight (consider the case of $\abs{\tens{T}} = 1$), but it might paint a slightly too pessimistic picture of the actual sparsity. Rather, we would like to follow~\cite{gillis10using} and relate the sparsity of the factors to the sparsity of the data. If $\tens{T}$ is an \byby{n}{m}{l} binary tensor, we define its \emph{sparsity} $s(\tens{T})$ as
\begin{equation}
  \label{eq:sparsity:def}
  s(\tens{T}) = 1 - \frac{\abs{\tens{T}}}{nml}\; .
\end{equation}
(The same notation is extended to matrices and vectors with one (respectively two) modes having dimension $1$.) Further, the rank-$r$ Boolean CP decomposition $(\matr{A}, \matr{B}, \matr{C})$ is \emph{reducible} if there exists an index $j\in\{1,\ldots,r\}$ such that
\begin{equation}
\bigvee_{i=1,\ldots,r}\vec{a}_{:i}\outprod\vec{b}_{:i}\outprod\vec{c}_{:i} = \bigvee_{\substack{i=1,\ldots,r\\i\neq j}} \vec{a}_{:i}\outprod\vec{b}_{:i}\outprod\vec{c}_{:i} \; .
\end{equation}
If the decomposition is not reducible, it is said to be \emph{irreducible}.

\begin{proposition}
  \label{prop:sparsity:rel}
  Let $\tens{T}$ be binary tensor that has a rank-$r$ irreducible (but not necessarily minimal) Boolean CP decomposition to factors $\matr{A}$, $\matr{B}$, and $\matr{C}$. Then
  \begin{equation}
    \label{eq:sparsity:rel}
    s(\matr{A}) + s(\matr{B}) + s(\matr{C}) \geq s(\tens{T})\; .
  \end{equation}
\end{proposition}

Our proof uses a similar technique as~\cite{gillis10using}. We will first proof the following special case.

\begin{lemma}
  \label{lemma:sparsity:rel}
  Proposition~\ref{prop:sparsity:rel} holds when $r=1$.
\end{lemma}

\begin{proof}
  Tensor $\tens{T}$ must be rank-1. Hence, $\abs{\tens{T}} = \abs{\matr{A}}\abs{\matr{B}}\abs{\matr{C}}$, or equivalently,
  $1 - s(\tens{T}) = (1 - s(\matr{A})) (1 - s(\matr{B})) (1 - s(\matr{C}))$. It follows that 
  \[
  \begin{split}
  s(\tens{T}) &= s(\matr{A}) + s(\matr{B}) + s(\matr{C}) - s(\matr{A}) s(\matr{B}) - s(\matr{A}) s(\matr{C}) 
   - s(\matr{B}) s(\matr{C}) + s(\matr{A}) s(\matr{B}) s(\matr{C}) \\
  &\leq s(\matr{A}) + s(\matr{B}) + s(\matr{C})\; ,   
  \end{split}
  \]
which holds as $s(\cdot)\in[0,1]$.
\end{proof}

\begin{proof}[Proof of Proposition~\ref{prop:sparsity:rel}]
  For this proof we need the concept of a \emph{residual tensor} $\tens{R}_k$. Let $\tens{R}_1 = \tens{T}$, and for $k=2,\ldots,r$, let $\tens{R}_k = \tens{R}_{k-1}\land \lnot\left(\vec{a}_{:k}\outprod\vec{b}_{:k}\outprod\vec{c}_{:k}\right)$, where the \textsc{and} and \textsc{not} are element-wise. Notice that, as the decomposition is irreducible, $s(\tens{R}_k) > s(\tens{R}_{k-1})$. 

  We now have for $k=1,\ldots,r$
  \[
  \begin{split}
  s(\vec{a}_{:k}) + s(\vec{b}_{:k}) + s(\vec{c}_{:k}) &\geq s(\vec{a}_{:k}\outprod\vec{b}_{:k}\outprod\vec{c}_{:k}) \geq s(\tens{R}_k) \geq s(\tens{T})
  \end{split}
  \]
  The first inequality follows from Lemma~\ref{lemma:sparsity:rel}, and the next ones from the fact that the Boolean semi-ring is anti-negative. The statement follows, as
  \[
  s(\matr{A}) = \frac{1}{r}\sum_{k=1}^r s(\vec{a}_{:k}),
  \]
  and similarly for $s(\matr{B})$ and $s(\matr{C})$.
\end{proof}

The above discussion strongly indicates that storing the data tensor in the factorized form can yield significant space savings, and in the worst case, should not increase the storage requirements too much. It seems reasonable to assume that the factorisation would also give benefits on the computations: after all, low-rank tensors should have more regular structure than high-rank tensors, and this regularity should help with the computations. Unfortunately, there does not seem to be much work towards that end. Bader and Kolda~\cite{bader2007efficient} study some operations on CP-decomposed tensors (under normal algebra), and while their results generally carry over to the Boolean algebra, the operations they study are not commonly seen in our framework.




%% file: related.tex
\section{Other Related Work}
\label{sec:relatedWork}
Tensors, as a way to represent multi-way relations, have been studied in the context of databases for a long time, although the term \emph{tensor} is not commonly used. Instead, terms like Data Cube~\cite{gray1996datacube} are used to refer to the data (and the associated operations, and the framework).

It is also not new to use the binary tensor representation of \RDF data to improve the processing of \SPARQL queries. For example, Atre et al.~\cite{atre2010matrix} propose a technique to effectively process join queries using binary tensors (referred to as \emph{Bit-cubes}). Subsequent work also extends the method to left outer joins~\cite{atre2013technique}.

Tensor decompositions have  been applied to \RDF data earlier. For example, Drumond et al.~\cite{drumond2012predicting} and Nickel et al.~\cite{nickel12factorizing} use the decompositions to predict missing or unobserved \RDF triples, while Erd\H{o}s et al.~\cite{erdos13discovering} use so-called Boolean Tucker3 decomposition to discover facts from ``surface'' $(s, p, o)$ triples. 


For relational algebra, the tensor relational model~\cite{kim2011approximate} is a framework that supports both relational algebraic operations for data manipulation and tensor algebraic operations for data analysis. Kim and Candan propose efficient ways to combine the costly tensor decomposition needed for data analysis together with join~\cite{kim2011approximate}, normalization~\cite{kim2012decomposition},  and union operations~\cite{kim2014pushing} for data manipulation. For the join for instance, the authors use non-negative tensor decompositions on the components to be joined in order to approximate the decomposition after the operation.

Bakibayev et al.~\cite{bakibayev12fdb} propose factorised relational databases that use compact factorized representations of data to improve query performance and to reduce redundancy.  The authors propose a specialized query engine to handle select-project-join queries on such data efficiently.




%% file: discussion.tex
\section{Future Work}
\label{sec:discussion}

In this paper we have presented a framework for \SPARQL queries as Boolean tensor operations. We believe that our framework can help with the analysis the \SPARQL queries, and with the development of new optimizations. 

The Boolean tensor rank measures the complexity of the tensor: the higher the rank, the less regularities the tensor has. It seems viable to use the tensor rank as a parameter of the complexity, much the same way as, say, the treewidth is used. Yet, we are unaware of any research towards that direction.

Most \SPARQL operations (especially the joins) can be expressed using the Khatri--Rao product. Unlike the normal matrix product, that has enjoyed on significant research interest over the years, the Khatri--Rao product is relatively unknown and unstudied. As it is the key for evaluating joins, insights on the computation of it can lead to direct real-world benefits. 

One can also ask the question the other way around, though. As the Khatri--Rao product (and the Boolean matrix product) can be expressed as \SPARQL queries, could this be used as a way to implement more efficient algorithms for (Boolean) tensor analysis. A relatively simple goal could be to use the data structures and indexing approaches employed by \RDF databases -- such as \textsc{rdf-3x}~\cite{neumann09rdf-3x} -- for more efficient data analysis algorithms. 

Going further, one can also consider implementing the whole data analysis algorithms on top of \RDF databases. Again, we argue that our tensor representation should help, giving a framework where many data analysis problems map easily. It is clear, though, that \SPARQL should be extended with new operations, should one want to implement more complicated data analysis directly on it (much the same ways as standard data mining methods can be integrated to relational databases; see, e.g.~\cite{netz01integrating,sarawagi00integrating}).

\todo{mention blank nodes}


%% file: conclusions.tex
\section{Conclusions}
\label{sec:conclusions}

We have presented a framework for \RDF and \SPARQL based on Boolean tensors. The framework allows us to cast \SPARQL queries as different types of matrix products, and use this formalisation to gain new insights and apply previously-defined techniques for their processing. Particularly, we showed how to count (and estimate) the cardinality of various \SPARQL join operations. We also briefly considered the Boolean CP decomposition as an option to find regularities from the data, and use these regularities for improved query processing and storage. Whether the decompositions would help on query processing is still wide open, but we did show two upper bounds for the space requirements of the decomposition. Overall, we see this work only as the first step, as we believe that the framework we presented here facilitates better analysis of \RDF\ and \SPARQL.
